\documentclass[11pt]{article}

\usepackage{newlfont}
\usepackage[below]{placeins}
\usepackage{flafter}
\usepackage{amsfonts}
\usepackage{amsmath}
\usepackage{amssymb}
\usepackage[american]{babel}
\usepackage{url}
\usepackage{graphicx} 
\usepackage{color}
\usepackage{transparent} 

\allowdisplaybreaks

\newcommand{\ket}[1]{\big| #1 \big\rangle}
\newcommand{\bra}[1]{\big\langle #1 \big|}
\newcommand{\braket}[2]{\big\langle #1 \big| #2 \big\rangle}             

\newtheorem{theorem}{Theorem}[section]
\newtheorem{lemma}[theorem]{Lemma}

\newtheorem{definition}[theorem]{Definition}

\newtheorem{remark}{Theorem}[section]
\newenvironment{proof}[1][Proof]{\begin{trivlist}
\item[\hskip \labelsep {\bfseries #1}]}{\end{trivlist}}

\newcommand{\qed}{\nobreak \ifvmode \relax \else
      \ifdim\lastskip<1.5em \hskip-\lastskip
      \hskip1.5em plus0em minus0.5em \fi \nobreak
      \vrule height0.75em width0.5em depth0.25em\fi}

\date{}

\begin{document}

\title{Coined Quantum Walks as Quantum Markov Chains}
\author{Renato Portugal$^1$\footnote{\tt portugal@lncc.br} \, and Etsuo Segawa$^2$\footnote{\tt e-segawa@m.tohoku.ac.jp}\\
{\small National Laboratory of Scientific Computing - LNCC} \\
{\small Av. Get\'{u}lio Vargas 333, Petr\'{o}polis, RJ, 25651-075, Brazil}\\
{\small $^2$Graduate school of Information Sciences, Tohoku University}\\
{\small Aoba, Sendai, 980-8579, Japan}
}

\maketitle

\begin{abstract}
We analyze the equivalence between discrete-time coined quantum walks and Szegedy's quantum walks. We characterize a class of flip-flop coined models with generalized Grover coin on a graph $\Gamma$ that can be directly converted into Szegedy's model on the subdivision graph of $\Gamma$ and we describe a method to convert one model into the other. This method improves previous results in literature that need to use the staggered model and the concept of line graph, which are avoided here. 
\end{abstract}

\section{Introduction}\label{sec:0}

Coined quantum walks (QWs) on graphs were defined in Ref.~\cite{Aharonov:2000} and have been extensively analyzed in literature~\cite{Ven12,Kon08,Kendon:2007,Portugal:Book,Manouchehri2014}. The coined model has an internal space, which determines the direction that a particle would take. The model's Hilbert space in this case is the tensor product of the internal space and the space associated with the graph.

Szegedy's model~\cite{Szegedy:2004}, on the other hand, does not have an internal state. This model provides a recipe to generate coinless discrete-time QWs on bipartite graphs using a Hilbert space associated with the graph only. Szegedy's model was used for the spatial search problem, that is, for finding the location of a marked vertex in a graph~\cite{Magniez:2011,Magniez:2012,Krovi:2010}, and for searching triangles~\cite{mss07}. 

The detailed connection between the coined and Szegedy's model has remained elusive for many years until 
Ref.~\cite{Por16} used the staggered QW model~\cite{PSFG15} as a bridge to describe under which conditions Szegedy's and coined QWs are equivalent. The method described in Ref.~\cite{Por16} employs the line graph of the graph on which the Szegedy's model is defined. In this work we describe a simpler method of obtaining the equivalence between Szegedy's and coined QWs using neither the staggered QW model nor the concept of line graphs. We characterize a class of flip-flop coined QWs on graph $\Gamma$ with generalized Grover coin that is equivalent to Szegedy's QWs on the subdivision graph of $\Gamma$ and describe how those coined QWs can be converted into Szegedy's model.

The structure of this paper is as follows. In Sec.~\ref{sec:MD}, we give the definition of the main concepts used in this work. In Sec.~\ref{sec:1}, we present our main result which is the method to convert coined QWs into Szegedy's QWs. In Sec.~\ref{sec_conlusions}, we draw our conclusions.

\section{Main Definitions}\label{sec:MD}

Let $\Gamma(V,E)$ be a multigraph with vertex set $V=V(\Gamma)$ and edge set $E=E(\Gamma)$ with cardinalities $|V|$ and $|E|$, respectively. 
%
We set ${\cal H}^{2|E|}$ as the Hilbert space whose computational basis is $\big\{\ket{v,j}: v\in V, \,0\le j<d_v\big\}$. 
We take a decomposition of ${\cal H}^{2|E|}$ by 
\begin{equation}\label{decom} 
        {\cal H}^{2|E|}=\bigoplus_{v\in V} \mathrm{span}\{ |v,j\rangle : j=0,1,\dots,d_v-1 \}, 
        \end{equation}
where $d_v$ is the degree of vertex $v$. 
\begin{definition}\label{def:nonregularQW}
The \textbf{flip-flop coined QW} on a multigraph $\Gamma(V,E)$ associated with Hilbert space ${\cal H}^{2\,|E|}$ is driven by a unitary operator the form of which is
\begin{equation}
	U \,=\, S\,C',
\end{equation}
where $C'$ is a direct sum of $|V|$ matrices under the decomposition of (\ref{decom}) with dimensions $d_1$, ..., $d_{|V|}$, 
and $S$ is the shift operator which permutes the vectors of the computational basis of ${\cal H}^{2\,|E|}$, 
\begin{equation}\label{def_S2}
	S \ket{v,j} = \ket{v',j'}, \,\forall v\in V, \,0\le j<d_v,
\end{equation}
where vertices $v$ and ${v'}$ are adjacent, label $j$ points from $v$ to ${v'}$, and label $j'$ points from ${v'}$ to $v$ and $j$, $j'$ lie on the same edge. 
\end{definition}
%
Notice that $\ket{v,j}$ is a notation for the basis vectors that cannot be written as $\ket{v}\otimes\ket{j}$ unless the multigraph is regular. 
We take the order of the basis vectors as $\ket{v_1,0}$, ...,  $\ket{v_1,d_1-1}$, $\ket{v_2,0}$, ..., $\ket{v_2,d_2-1}$, etc., so that $C'$ will have a block diagonal form with $|V|$ matrices.

We set $A=A(\Gamma)$ as the set of symmetric arcs induced by $E(\Gamma)$, that is, $E(\Gamma)=\{ \{a,\bar{a}\}: a\in A \}$. 
Here we denote the origin and terminal vertices of $a\in A(\Gamma)$ by $o(a),t(a)\in V$, and the inverse arc of $a$ by $\bar{a}$. 
\noindent \\
\begin{remark}\label{rem1}
It holds that $|A|=|\bigcup_{v\in V}\{(v,j): j\in \{1,\dots,d_v\} \}|$ since 
a one-to-one correspondence between them is $(v,j)\leftrightarrow a$, where $o(a)=v$ with $e:=\{a,\bar{a}\}\in E$, and $j$ lies on $e$. 
We define bijection $\xi:A\to \bigcup_{v\in V}\{(v,j): j\in \{1,\dots,d_v\} \}$ by
	\begin{equation}\label{xi}
        \xi(a)=(o(a), j), \mathrm{\;where\;} j \mathrm{\;lies\; on\;} \{a,\bar{a}\}.  
        \end{equation} 
\end{remark}
Let $\mathcal{H}_A$ be the Hilbert space whose computational basis is $\{|a\rangle: a\in A\}$. 
We define the unitary map $\mathcal{V}_\xi: \mathcal{H}_A \to \mathcal{H}^{2|E|}$ by 
	\[ \mathcal{V}_\xi|a\rangle=| o(a),j \rangle, \]
where $j$ lies on $\{a,\bar{a}\}$. 
The shift operator is expressed by $\mathcal{V}_\xi^{-1} S\mathcal{V}_\xi |a\rangle=|\bar{a}\rangle$. 
For $v\in V$, define the subspace $\mathcal{H}_v\subset \mathcal{H}_A$ by $\mathcal{H}_v=\mathrm{span}\{|a\rangle: o(a)=v\}$ which is isomorphic to 
$\mathrm{span}\{|v,j\rangle: j=0,\dots, d_v-1\}\subset \mathcal{H}^{2|E|}$. 
Thus if $C'=\sum_{v\in V} \oplus C_v$, then the coin operator is expressed by 
	\[ \mathcal{V}_\xi^{-1} C'\mathcal{V}_\xi|a\rangle=\sum_{b\in \{b\in A:o(a)=o(b)\}} (C_{o(a)})_{\xi^{-1}(b),\xi^{-1}(a)} |b\rangle. \] 
\noindent \\
\quad Let us define the QW model known as Szegedy's model~\cite{Szegedy:2004}. Consider a connected bipartite graph $\Gamma(X,Y,E')$, where $X,Y$ are disjoint sets of vertices and $E'$ is the set of non-directed edges. Let 
\begin{equation}\label{biadmatrix}
		\left(\begin{array}[]{cc}
		  0 & M \\
		M^T & 0
	\end{array}\right)
\end{equation}
be the biadjacency matrix of $\Gamma(X,Y,E')$, that is, $(M)_{x,y}=1$ if $\{x,y\}\in E$, $(M)_{x,y}=0$ otherwise. Here $T$ is the transpose operator. 
Using $M$, define $P$ as a probabilistic map from $X$ to $Y$ with entries $p_{xy}\ge 0$, that is, $(M)_{x,y}=0\Rightarrow p_{xy}=0$. 
Using $M^T$, define $Q$  as a probabilistic map from $Y$ to $X$ with entries $q_{yx}\ge 0$, that is, $(M^T)_{y,x}=0\Rightarrow q_{yx}=0$. 
If $P$ is an $m\times n$ matrix, $Q$ will be an $n\times m$ matrix. Both are right-stochastic, that is, each row sums to $1$. 
Let $\mathcal{H}^{mn}=\mathcal{H}^{m}\otimes \mathcal{H}^{n}$ be the Hilbert space whose canonical basis is $\{|x,y\rangle:=|x\rangle \otimes |y\rangle; x\in X,\;y\in Y\}$. 
Using $P$ and $Q$, it is possible to define unit vectors on $\mathcal{H}^{mn}$ for given $x\in X$ and $y\in Y$, 
\begin{eqnarray}
  \ket{\phi_x} &=&  \sum_{y\in Y} \sqrt{p_{x y}}\,\textrm{e}^{i\theta_{xy}} \, \ket{x,y}, \label{ht_phi_x} \\
  \ket{\psi_y}  &=&  \sum_{x\in X} \sqrt{q_{y x}}\,\textrm{e}^{i\theta'_{xy}} \, \ket{x,y}, \label{ht_psi_y}
\end{eqnarray}
that have the following properties: $\braket{\phi_x}{\phi_{x'}}=\delta_{xx'}$ and $\braket{\psi_y}{\psi_{y'}}=\delta_{yy'}$. 
Here $\theta_{xy},\theta'_{xy}\in \mathbb{R}$. 
In Szegedy's original definition, $\theta_{xy}=\theta'_{xy}=0$. We call \textbf{extended Szegedy's QW} the version that allows nonzero angles.\footnote{Another extended version of  Szegedy's QW model can be seen in~\cite{SS}.}

\begin{definition}\label{def:SzegedyQW}
\textbf{Szegedy's QW} on a bipartite graph $\Gamma(X,Y,E)$ with biadjacent matrix (\ref{biadmatrix}) is defined on a Hilbert space ${\cal H}^{m n} = {\cal H}^{m}\otimes {\cal H}^{n} $, where $ m = | X |$ and $n = | Y | $, the computational basis of which is $ \big \{\ket {x, y}: x \in X, y \in Y \big \} $.
The QW is driven by the unitary operator
\begin{equation}\label{ht_U_ev}
    W \,=\, R_1 \, R_0,
\end{equation}
where
\begin{eqnarray}
  R_0 &=& 2\sum_{x\in X} \ket{\phi_x}\bra{\phi_x} - I, \label{ht_RA}\\
  R_1 &=& 2\sum_{y\in Y} \ket{\psi_y}\bra{\psi_y} - I. \label{ht_RB}
\end{eqnarray}
\end{definition}
Notice that operators $R_0$ and $R_1$ are unitary and Hermitian ($R_0^2=R_1^2=I$).

Let us define the notion of generalized Grover operator.
\begin{definition}\label{def:orthrefl}
Let $\mathcal{H}^{d_v}\subset \mathcal{H}^{2|E|}$ ($v\in V$) be the subspace $\mathrm{span}\{|v,j\rangle: j=0,\dots, d_v-1\}$. 
A generalized Grover operator ${\cal G}_v$ on ${\cal H}^{d_v}$ has the form 
\begin{equation}
	{\cal G}_v\,=\, 2\ket{\psi_v}\bra{\psi_v} - I_{\mathcal{H}^{d_v}},
\end{equation}
where $\ket{\psi_v} \in \mathcal{H}^{d_v}$ is a unit vector. 
\end{definition}
The Grover operator is obtained when $\langle v,j|\psi_v\rangle=1/{\sqrt d_v}$ for all $j\in \{0,\dots, d_v-1\}$. 
We put $|\alpha^{(v)}\rangle:= \mathcal{V}_\xi^{-1}|\psi_v\rangle\in \mathcal{H}_A$. 
Thus the generalized Grover operator on $\mathcal{H}_v\subset \mathcal{H}_A$ is $C_v:=\mathcal{V}_\xi {\cal G}_v \mathcal{V}_\xi^{-1}$, expressed by 
	\[ C_v= 2|\alpha^{(v)}\rangle \langle \alpha^{(v)}|-I_{\mathcal{H}_v}. \]        

\begin{definition}\label{def:subdivisiongraph}
The \textbf{subdivision graph} $S(\Gamma)$ of a multigraph $\Gamma=(V,E)$ is defined as follows: 
	\begin{align*} 
	V(S(\Gamma)) &= V(\Gamma) \cup E(\Gamma), \\
        E(S(\Gamma)) &= \{ \{v,e\}: v\in V(\Gamma), \; e\in E(\Gamma),\; v\in e \}. 
        \end{align*}
\end{definition}
That means that $S(\Gamma)$ is the graph obtained from $\Gamma$ by adding a vertex in the middle of each edge of $\Gamma$.
\noindent \\
\begin{remark}\label{rem4}
The map $\eta: A(\Gamma)\to E(S(\Gamma))$ such that $\eta(a)=\{ o(a), \{a,\bar{a}\}\}$ is bijection since $S(\Gamma)$ is a simple graph.  
\end{remark}
\noindent \\
See Fig.~\ref{fig:A}, for the equivalence relations of $\{ |v,j\rangle: v\in V,\;0\leq j\leq d_v \}$, $A(\Gamma)$ and $E(S(\Gamma))$. 
\begin{figure}[htbp]
\begin{center}
	\includegraphics[width=100mm]{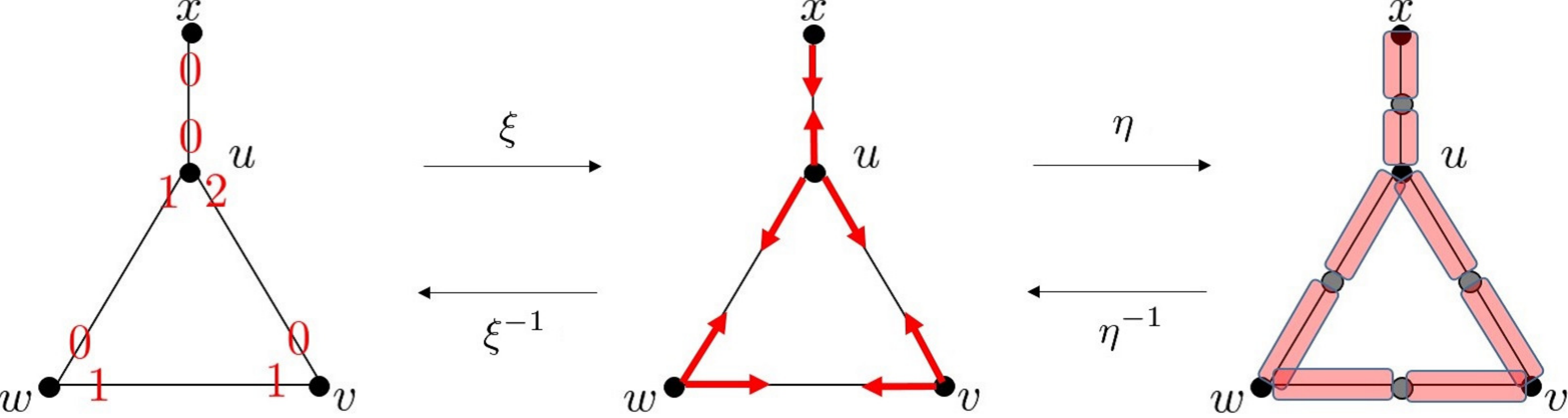}
\end{center}
\caption{The equivalence relations of $\{ |v,j\rangle: v\in V,\;0\leq j\leq d_v \}$, $A(\Gamma)$ and $E(S(\Gamma))$. }
\label{fig:A}
\end{figure}

\quad Proposition~7.2.2 of Ref.~\cite{GY05} shows that the subdivision graph $S(\Gamma)$ of a multigraph $\Gamma$ is a bipartite graph $\Gamma'=(X,Y,E')$. 
Set $X$ comprises the vertices $v\in V$ and set $Y$ comprises the new vertices so that $|Y|=|E|$. 
The cardinality of $E'$ is $2|E|$. 

\section{Main Results}\label{sec:1}
Consider the following lemma. 
\begin{lemma}\label{rem2}
Let $\mathcal{H}_E \subset \mathcal{H}^{nm}$ be the subset whose computational set is $\{ \ket{x}\otimes \ket{y}: \{x,y\}\in E(\Gamma) \}$.
It holds that 
	\[ W(\mathcal{H}_E)=\mathcal{H}_E,\;\;W|_{{\mathcal{H}_E}^\perp}=1, \]
which means a matrix representation of $W$ is 
	\begin{equation}
	W \,=\, \left[ \begin{array}{cccc}
	W|_{\mathcal{H}_E} & 0 & \cdots & 0 \\
	0 & 1 & \cdots & 0 \\
	\vdots & \vdots  & \ddots & \vdots  \\
	0 & 0 &\cdots & 1 \end{array} \right],
	\end{equation}
\end{lemma}
\begin{proof}
Let $\Pi$ be the orthogonal projection onto $\mathcal{H}_E$. Put $\psi=\psi_1+\psi_2$ with $\psi_1\in \mathcal{H}_E$ and $\psi_2\in {\mathcal{H}_E}^\perp$. 
Then we have $\Pi W\psi=\Pi W(\psi_1+\psi_2) = W\psi_1$ since $R_0\psi_2=R_1\psi_2=-\psi_2$. On the other hand, $W\Pi\psi=W\psi_1$. 
So we have $\Pi W=W\Pi$, which implies $W(\mathcal{H}_E)=\mathcal{H}_E$. Indeed, for any $\ket{f}\in \mathcal{H}_E, \ket{g}\in {\mathcal{H}_E}^\perp$, 
$\braket{g}{Wf}=\braket{W^*g}{f}=\braket{g}{f}=0$, which implies $W(\mathcal{H}_E)\subset \mathcal{H}_E$. 
On the other hand, since $W$ is a bijection onto $\mathcal{H}^{nm}$, for any $f\in\mathcal{H}^{nm}$, there uniquely exists $\ket{g}\in \mathcal{H}^{nm}$
such that $\ket{f}=W\ket{g}$. Using this $\Pi \ket{f}=\Pi W\ket{g}=W\Pi \ket{g}\in W(\mathcal{H}_E)$. 
\end{proof}

Due to the isomorphism between $\mathcal{H}_{A}$ and $\mathcal{H}^{2|E|}$ in (\ref{xi}), from now on we identify $\mathcal{H}^{2|E|}$ with $\mathcal{H}_A$, 
$\mathcal{V}_\xi C'\mathcal{V}_\xi^{-1}$ with $C'$, and $\mathcal{V}_\xi S\mathcal{V}_\xi^{-1}$ with $S$. 
The main result of this work is:
\begin{theorem}\label{theo1}
Let $\mathcal{U}_\eta : \mathcal{H}_{A(\Gamma)}\to \mathcal{H}_{E(S(\Gamma))}$ be the unitary representation of the bijection map 
$\eta: A(\Gamma)\to E(S(\Gamma))$ in Remark~\ref{rem4}, that is, 
\begin{equation}\label{mapUeta}
\mathcal{U}_\eta\ket{a} = \ket{o(a)}\otimes \ket{\{a,\bar{a}\}}.
\end{equation}
A flip-flop coined QW on a multigraph $\Gamma=(V,E)$ such that $C'=\bigoplus_{v\in V} C_v$ and each $C_v$ is a generalized Grover operator can be cast into the extended Szegedy's model on the subdivision graph $S(\Gamma)$, 
that is, 
	\[ U=\mathcal{U}_\eta^{-1} W|_{\mathcal{H}_{E(S(\Gamma))}} \mathcal{U}_\eta, \]
where $U: \mathcal{H}_{A(\Gamma)}\to \mathcal{H}_{A(\Gamma)}$ is the time evolution of a flip-flop coined QW on $\Gamma$ with the generalized Grover coin and 
$W: \mathcal{H}^{|V(\Gamma)|}\otimes \mathcal{H}^{|E(\Gamma)|}\to \mathcal{H}^{|V(\Gamma)|}\otimes \mathcal{H}^{|E(\Gamma)|}$ 
is the time evolution of an extended Szegedy walk on $S(\Gamma)$. 
\end{theorem}
\begin{proof}
Since $C_v$ is a generalized Grover operator, it can be written as
\begin{equation}\label{C}
	C_v\,=\,2\,\ket{\alpha^{(v)}}\bra{\alpha^{(v)}}-I_{\mathcal{H}_v},
\end{equation}
where $\mathcal{H}_v$ is the subspace spanned by $\{\ket{a}:o(a)=v\}$ and $\ket{\alpha^{(v)}}$ is a unit vector on $\mathcal{H}_v$. 
We put 
	\begin{equation}\ket{\alpha^{(v)}}\,=\, \sum_{a:o(a)=v} \alpha_a \ket{a} \label{alpha_v} \end{equation}
with $\sum_{a:o(a)=v}|\alpha_a|^2=1$. 
By the assumption of a generalized Grover operator, we have $\alpha_a:=\langle a|\alpha^{(v)}\rangle\neq 0$ for every $a\in A(\Gamma)$ with $o(a)=v$. 
Since the coin $C'$ is the direct sum of generalized Grover operators $\{C_v\}_{v\in V(\Gamma)}$, then the coin $C'$ is rewritten as 
\begin{equation}
	C'\,=\,2\sum_{v\in V(\Gamma)}\ket{\alpha^{(v)}}\bra{\alpha^{(v)}}-I_{\mathcal{H}_{A(\Gamma)}}. 
\end{equation}
Therefore $\mathrm{spec}(C')=\{\pm 1\}$ and 
	\begin{align*}
        \mathrm{ker}(1-C') &= \mathrm{span}\{|\alpha^{(v)}\rangle: v\in V\}. 
        \end{align*}
We call the above LHS $(+1)$-eigenspace of $C'$. 
Since $S|a\rangle=|\bar{a}\rangle$ and $S|\bar{a}\rangle=|a\rangle$, then $\mathrm{span}\{|a\rangle, |\bar{a}\rangle\}\subset \mathcal{H}_A$ is invariant under the action of $S$, and $S$ acts as 
	\[ S\cong \begin{bmatrix}  0 & 1 \\ 1 & 0 \end{bmatrix} \]
in this invariant subspace. Therefore $\mathrm{spec}(S)=\{\pm 1\}$, and 
	\begin{align*}
	\mathrm{ker}(1-S) &= \mathrm{span}\{|\beta^{(e)}\rangle: e\in E\}, \\
        \mathrm{ker}(1+S) &= \mathrm{span}\{|\gamma^{(e)}\rangle: e\in E\}.
	\end{align*}
Here 
\begin{align}
	\ket{\beta^{(e)}}\,=\, \frac{\ket{a}+\ket{\bar{a}}}{\sqrt 2}, \;\; \ket{\gamma^{(e)}}\,=\, \frac{\ket{a}-\ket{\bar{a}}}{\sqrt 2},\label{evS}
\end{align}
where $e=\{a,\bar{a}\}\in E(\Gamma)$. 
The $\mathrm{ker}(1-S)$ is the $(+1)$-eigenspace of $S$. 
Thus $S$ is expressed by 
\begin{align}
	S &= \sum_{e\in E(\Gamma)}\ket{\beta^{(e)}}\bra{\beta^{(e)}}-\ket{\gamma^{(e)}}\bra{\gamma^{(e)}} \notag \\
          &= 2\sum_{e\in E(\Gamma)}\ket{\beta^{(e)}}\bra{\beta^{(e)}}-I_{\mathcal{H}_{A(\Gamma)}}.
\end{align}

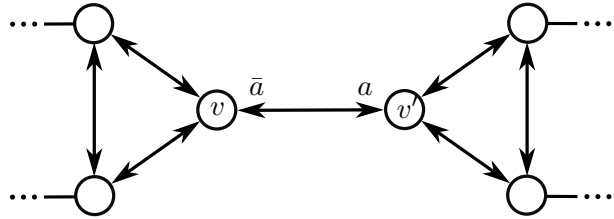
\begin{figure}[h!] 
\centering 
\def\svgwidth{8cm} 
\begingroup%
  \makeatletter%
  \providecommand\color[2][]{%
    \errmessage{(Inkscape) Color is used for the text in Inkscape, but the package 'color.sty' is not loaded}%
    \renewcommand\color[2][]{}%
  }%
  \providecommand\transparent[1]{%
    \errmessage{(Inkscape) Transparency is used (non-zero) for the text in Inkscape, but the package 'transparent.sty' is not loaded}%
    \renewcommand\transparent[1]{}%
  }%
  \providecommand\rotatebox[2]{#2}%
  \ifx\svgwidth\undefined%
    \setlength{\unitlength}{445.81422596bp}%
    \ifx\svgscale\undefined%
      \relax%
    \else%
      \setlength{\unitlength}{\unitlength * \real{\svgscale}}%
    \fi%
  \else%
    \setlength{\unitlength}{\svgwidth}%
  \fi%
  \global\let\svgwidth\undefined%
  \global\let\svgscale\undefined%
  \makeatother%
  \begin{picture}(1,0.35531431)%
    \put(0,0){\includegraphics[width=\unitlength]{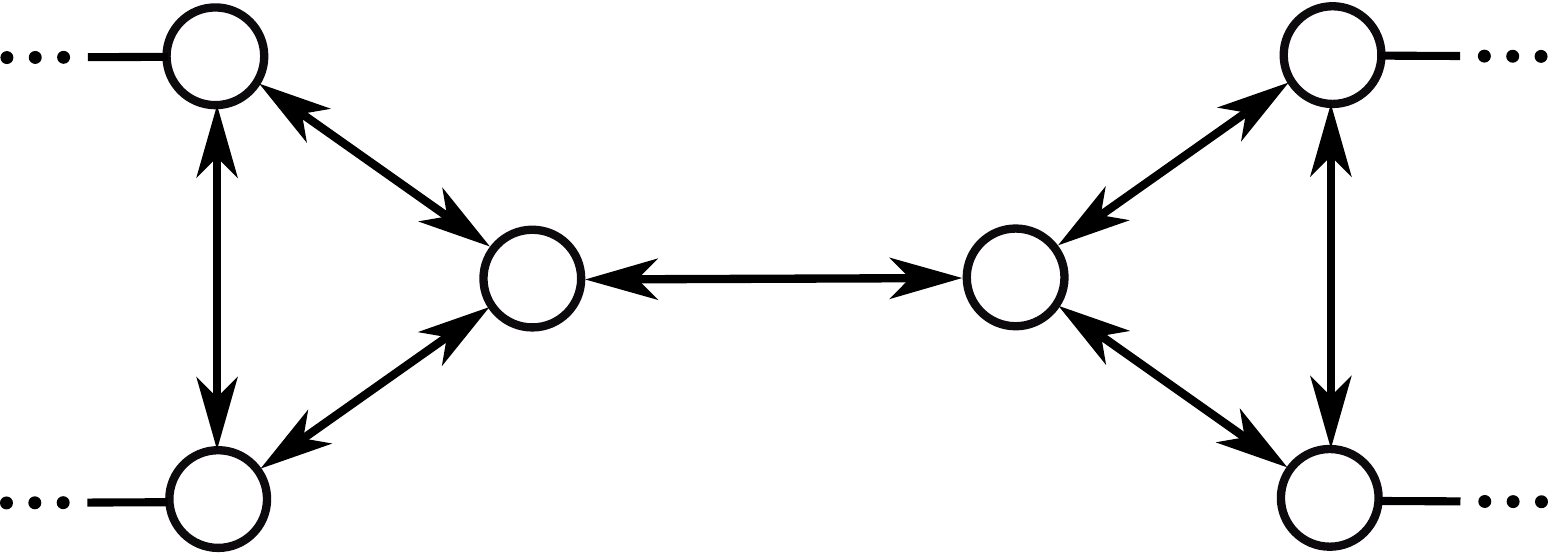}}%
    \put(-0.08617618,0.34213219){\color[rgb]{0,0,0}\makebox(0,0)[lb]{\smash{}}}%
    \put(0.50669734,-0.22616816){\color[rgb]{0,0,0}\makebox(0,0)[b]{\smash{}}}%
    \put(0.25816196,-0.00850364){\color[rgb]{0,0,0}\makebox(0,0)[lt]{\begin{minipage}{0.18721396\unitlength}\centering \end{minipage}}}%
    \put(0.2536136,0.1418708){\color[rgb]{0,0,0}\makebox(0,0)[b]{\smash{}}}%
    \put(0.51537649,-0.14114635){\color[rgb]{0,0,0}\makebox(0,0)[lt]{\begin{minipage}{0.09284753\unitlength}\centering \end{minipage}}}%
    \put(0.20956635,-0.24522317){\color[rgb]{0,0,0}\makebox(0,0)[lt]{\begin{minipage}{0.10441216\unitlength}\centering \end{minipage}}}%
    \put(0.06718612,0.0316273){\color[rgb]{0,0,0}\makebox(0,0)[b]{\smash{}}}%
    \put(0.62904765,0.16647129){\color[rgb]{0,0,0}\makebox(0,0)[b]{\smash{}}}%
    \put(0.39778787,0.20399335){\color[rgb]{0,0,0}\makebox(0,0)[lb]{\smash{$\bar{a}$}}}%
    \put(0.57750194,0.20343417){\color[rgb]{0,0,0}\makebox(0,0)[lb]{\smash{$a$}}}%
    \put(0.33268244,0.17108898){\color[rgb]{0,0,0}\makebox(0,0)[lb]{\smash{$v$}}}%
    \put(0.64307154,0.16586854){\color[rgb]{0,0,0}\makebox(0,0)[lb]{\smash{$v'$}}}%
  \end{picture}%
\endgroup%

\caption{Example of a graph depicting two generic vertices $v$ and $v'$. The coin direction $j$ represented by arc $a$ points from $v$ to $v'$ and $j'$ represented by arc $\bar{a}$ points from $v'$ to $v$. } 
\label{fig:graph1}
\end{figure}

Consider the subdivision graph $S(\Gamma)$. $S(\Gamma)$ is obtained from $\Gamma$ by adding a new vertex in the middle of each edge $e\in E$. 
If $e$ contains arcs $a$ and $\bar{a}$ as depicted in Fig.~\ref{fig:graph1}, 
the label for the new vertex is $\{a,\bar{a}\}$ or equivalently $\{\bar{a},a\}$. 
We consider identical the labels $\{a,\bar{a}\}$ and $\{\bar{a},a\}$. The new vertex is depicted in Fig.~\ref{fig:graph2}. 
Remark~\ref{rem4} implies $\mathcal{H}_{A(\Gamma)}\cong \mathcal{H}_{E(S(\Gamma))}\subset \mathcal{H}^{|V(\Gamma)|}\otimes \mathcal{H}^{|E(\Gamma)|}$. 
The goal now is to define a Szegedy model on the subdivision graph $S(\Gamma)$. 
To this end, consider Hilbert space ${\cal H}^{|V(\Gamma)|}\otimes{\cal H}^{|E(\Gamma)|}$ as the total state space of our desired Szegedy's model.
The computational basis uses the following notation: 
The first $2|E|$ vectors are given by $\ket{v}\otimes\ket{e}$, where $v\in V(\Gamma)$ is in the end of $e\in E(\Gamma)$. 
We will consider identical the vectors $\ket{e}=\ket{\{a,\bar{a}\}}$ and $\ket{\{\bar{a},a\}}$ with $e=\{a,\bar{a}\}$. 
The remaining vectors are given by $\ket{v}\otimes\ket{e}$, where $v$ is not the ends of $e$. 

\begin{figure}[h!] 
\centering 
\def\svgwidth{8cm} 
\begingroup%
  \makeatletter%
  \providecommand\color[2][]{%
    \errmessage{(Inkscape) Color is used for the text in Inkscape, but the package 'color.sty' is not loaded}%
    \renewcommand\color[2][]{}%
  }%
  \providecommand\transparent[1]{%
    \errmessage{(Inkscape) Transparency is used (non-zero) for the text in Inkscape, but the package 'transparent.sty' is not loaded}%
    \renewcommand\transparent[1]{}%
  }%
  \providecommand\rotatebox[2]{#2}%
  \ifx\svgwidth\undefined%
    \setlength{\unitlength}{468.89639396bp}%
    \ifx\svgscale\undefined%
      \relax%
    \else%
      \setlength{\unitlength}{\unitlength * \real{\svgscale}}%
    \fi%
  \else%
    \setlength{\unitlength}{\svgwidth}%
  \fi%
  \global\let\svgwidth\undefined%
  \global\let\svgscale\undefined%
  \makeatother%
  \begin{picture}(1,0.33795675)%
    \put(0,0){\includegraphics[width=\unitlength]{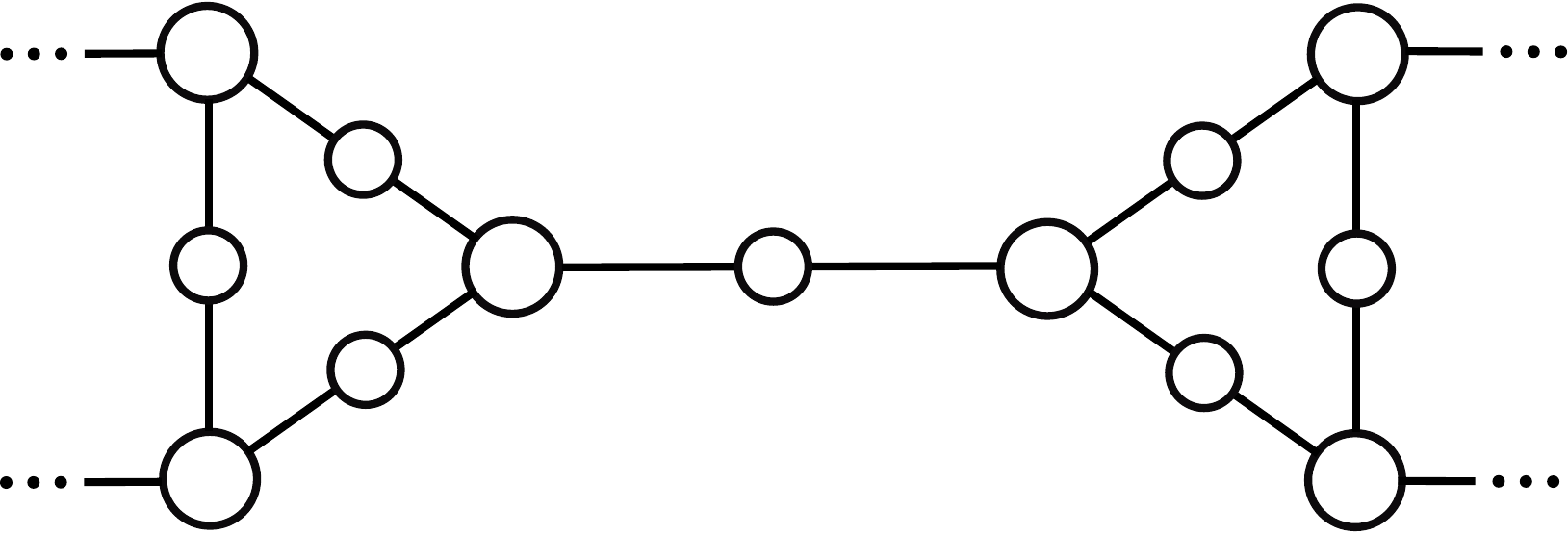}}%
    \put(-0.00345186,0.32615911){\color[rgb]{0,0,0}\makebox(0,0)[lb]{\smash{}}}%
    \put(0.56023652,-0.21416575){\color[rgb]{0,0,0}\makebox(0,0)[b]{\smash{}}}%
    \put(0.32393568,-0.00721611){\color[rgb]{0,0,0}\makebox(0,0)[lt]{\begin{minipage}{0.17799806\unitlength}\centering \end{minipage}}}%
    \put(0.31961122,0.13575591){\color[rgb]{0,0,0}\makebox(0,0)[b]{\smash{}}}%
    \put(0.56848842,-0.13332928){\color[rgb]{0,0,0}\makebox(0,0)[lt]{\begin{minipage}{0.08827696\unitlength}\centering \end{minipage}}}%
    \put(0.27773227,-0.23228275){\color[rgb]{0,0,0}\makebox(0,0)[lt]{\begin{minipage}{0.0992723\unitlength}\centering \end{minipage}}}%
    \put(0.14236093,0.03093932){\color[rgb]{0,0,0}\makebox(0,0)[b]{\smash{}}}%
    \put(0.67656394,0.1591454){\color[rgb]{0,0,0}\makebox(0,0)[b]{\smash{}}}%
    \put(0.38839534,0.18892779){\color[rgb]{0,0,0}\makebox(0,0)[b]{\smash{}}}%
    \put(0.6045509,0.0848115){\color[rgb]{0,0,0}\makebox(0,0)[lb]{\smash{}}}%
    \put(0.43823228,0.10){\color[rgb]{0,0,0}\makebox(0,0)[lb]{\smash{$\{a,\bar{a}\}$}}}%
    \put(0.37445808,0.18834854){\color[rgb]{0,0,0}\makebox(0,0)[lb]{\smash{$j$}}}%
    \put(0.31695216,0.16139805){\color[rgb]{0,0,0}\makebox(0,0)[lb]{\smash{$v$}}}%
    \put(0.65364443,0.15459466){\color[rgb]{0,0,0}\makebox(0,0)[lb]{\smash{$v'$}}}%
    \put(0.60086086,0.18837492){\color[rgb]{0,0,0}\makebox(0,0)[lb]{\smash{$j'$}}}%
  \end{picture}%
\endgroup%
\caption{Subdivision graph depicting the label $\{a,\bar{a}\}$ for the new vertex placed between vertices $v$ and $v'$ with coin directions $j$ and $j'$. } 
\label{fig:graph2}
\end{figure}
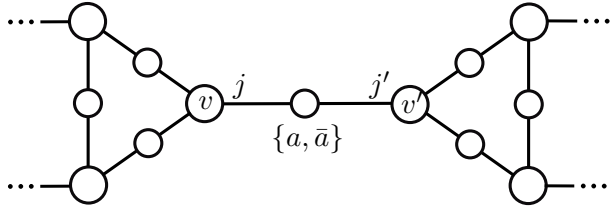
Vectors $\ket{\alpha^{(v)}}$ are given in terms of the computational basis $\{\ket{a}:a\in A\}$ of ${\cal H}_A$.
Define vectors $\ket{\phi^{(v)}}$ in  $\mathcal{H}_{E(S(\Gamma))}\subset {\cal H}^{|V(\Gamma)|}\otimes{\cal H}^{|E(\Gamma)|}$ using $\ket{\alpha^{(v)}}$ given by (\ref{alpha_v}) 
and the one-to-one map given by~(\ref{mapUeta}); 
we replace $\ket{a}\in \mathcal{H}_{A(\Gamma)}$ by $\ket{o(a)}\otimes\ket{\{a,\bar{a}'\}}\in \mathcal{H}_{E(S(\Gamma))}$ obtaining 
\begin{equation}\label{phi_v}
	\ket{\phi^{(v)}}\,=\, \sum_{a:o(a)=v} \alpha_a \ket{v}\otimes\ket{\{a,\bar{a}\}},
\end{equation}
and the same for vectors $\ket{a}$ in $\ket{\beta^{(e)}}$ given by (\ref{evS}), they are also replaced by $\ket{o(a)}\otimes\ket{\{a,\bar{a}\}}$ 
obtaining vectors $\ket{\psi^{(e)}}$, which are defined as
\begin{equation}\label{psi_e}
	\ket{\psi^{(e)}}\,=\, \frac{\ket{o(a)}\otimes\ket{\{a,\bar{a}\}}+\ket{t(a)}\otimes\ket{\{a,\bar{a}\}}}{\sqrt 2},
\end{equation}
where $e=\{a,\bar{a}\}$. 
Notice that $\ket{\psi^{(e)}}$ can be factorized because the $\ket{\{a,\bar{a}\}}$ and $\ket{\{\bar{a},a\}}$ are identical. 

Now we define an operator on ${\cal H}^{|V(\Gamma)|}\otimes{\cal H}^{|E(\Gamma)|}$ that is expected to be a well defined evolution operator of the Szegedy's model by $W=R_1 R_0$, where
\begin{eqnarray}
	R_0 &=& 2\sum_{v\in V(\Gamma)}\ket{\phi^{(v)}}\bra{\phi^{(v)}}-I,\\
	R_1 &=& 2\sum_{e\in E(\Gamma)}\ket{\psi^{(e)}}\bra{\psi^{(e)}}-I.
\end{eqnarray}
Using~(\ref{mapUeta}), we obtain $\mathcal{U}_\eta\ket{\alpha^{(v)}}=\ket{\phi^{(v)}}$ and $\mathcal{U}_\eta\ket{\beta^{(e)}}=\ket{\psi^{(e)}}$. 
Thus 
	\[ \mathcal{U}_\eta \,C'\, \mathcal{U}_\eta^{-1}=2 \sum_{v\in V}\ket{\phi^{(v)}}\bra{\phi^{(v)}}-I_{\mathcal{H}_{E(S(\Gamma))}}=\left(2 \sum_{v\in V}\ket{\phi^{(v)}}\bra{\phi^{(v)}}-I\right)\bigg|_{\mathcal{H}_{E(S(\Gamma))} } \] 
and 	\[ \mathcal{U}_\eta \,S\, \mathcal{U}_\eta^{-1}=2 \sum_{e\in E}\ket{\psi^{(e)}}\bra{\psi^{(e)}}-I_{\mathcal{H}_{E(S(\Gamma))}}=\left(2 \sum_{e\in E}\ket{\psi^{(e)}}\bra{\psi^{(e)}}-I\right)\bigg|_{\mathcal{H}_{E(S(\Gamma))} }. \]
Thus we obtain $U=\mathcal{U}_\eta^{-1}W|_{\mathcal{H}_{E(S(\Gamma))}}\mathcal{U}_\eta$. 

Next, let us show that this operator $W$ is a well-defined Szegedy evolution operator. 
First we check that $W$ is restricted to $\mathcal{H}_{E(S(\Gamma))}$, 
which is spanned by the first $2|E|$ computational vectors of $\mathcal{H}^{|V|}\otimes \mathcal{H}^{|E|}$. 
Using $\ket{\phi^{(v)}}$ and $\ket{\psi^{(e)}}$, 
we define matrices $P$ and $Q$, whose entries are $p_{v,e}=\left|\alpha_e\right|^2$ and $q_{e,v}=q_{e,v'}=1/2$, respectively. 
$P$ and $Q$ are right transition matrices. In fact, $\sum_e p_{v,e}=1$, $\forall v\in V$ because $\ket{\alpha^{(v)}}$ has unit $\ell^2$-norm and $q_{e,v}+q_{e,v'}=1$, $\forall e\in E$.  
Let $P',Q'$ be the matrices obtained from $P,Q$ by replacing the nonzero entries by $1$. 
We also have to show that $Q'^{\textrm{T}}=P'$ (see Definition \ref{def:SzegedyQW}), or equivalently $p'_{v,e}=1$ $\Leftrightarrow$ $q'_{e,v}=1$. 
We have the following equivalent relations : 
	\begin{align*}
        p'_{v,e}=1 & \Leftrightarrow \{v,e\}\in E(S(G)) 
        	     \Leftrightarrow e\in \{a,\bar{a}\} \mathrm{\;with\;} v\in\{ o(a),t(a) \} \\
        	   & \Leftrightarrow \ket{v}\otimes \ket{\{a,\bar{a}\}}\in \mathrm{Computational \;basis\; of\;}\mathcal{H}_{E(S(\Gamma))} \mathrm{\;with\;} v\in\{ o(a),t(a) \} \\
                   & \Leftrightarrow q_{e,v}=1/\sqrt{2} 
                     \Leftrightarrow q'_{e,v}=1. 
        \end{align*} 
The first equivalence of ``$\Leftarrow$'' follows from $\alpha_a \neq 0$ for all $a\in \{a'\in A(\Gamma): o(a')=v\}$, which is our assumption. 
The second follows from the last equivalence and the definition of $\psi^{(e)}$ presented in~(\ref{psi_e}). Thus we obtain $Q'^{\textrm{T}}=P'$. 

Finally, we check the well-definedness of the remaining space $\mathcal{H}_{E(S(\Gamma))}^\perp$ spanned by the $|V||E|-2|E|$ vectors. 
For any $\ket{f}\in \mathcal{H}_{E(S(\Gamma))}^\perp$, $\ket{\phi^{(v)}}$ and $\ket{\psi^{(e)}}$ are orthogonal to $\ket{f}$ since all the computational basis of $\mathcal{H}_{E(S(\Gamma))}^\perp$ are 
orthogonal to $\ket{\phi^{(v)}}$ and $\ket{\psi^{(e)}}$ by the definition. 
Thus $R_0\ket{f}=R_1\ket{f}=-\ket{f}$, which implies $W\ket{f}=\ket{f}$. 
Therefore we have $W|_{\mathcal{H}_{E(S(\Gamma))}^\perp}=I_{\mathcal{H}_{E(S(\Gamma))}}$. 
By Lemma~\ref{rem2}, the walk restricted to $\mathcal{H}_{E(S(\Gamma))}^\perp$ is well-defined.
Taken all together, the walk, whose time evolution is driven by $W$, is a well-defined Szegedy walk.
\end{proof}

In Theorem~\ref{theo1}, if $C_v$ for some vertex $v$ has more the one $(+1)$-eigenvector, $C_v$ is similar to the direct sum of smaller matrices. In this case, the graph on which Szegedy's model takes place is not the subdivision graph of the graph on which the coined model takes place. This case was addressed in Ref.~\cite{Por16}.

Finally, we provide an example in Figure \ref{fig:example}, which shows an application of Theorem~\ref{theo1}. In this figure we take \[ \sum_{y\in V(\Gamma)} \oplus C_y'=R_0\] with $C'_y=\mathcal{U}_\eta C_y\mathcal{U}_\eta^{-1}$, $(y\in \{u,v,w,x\})$, 
and \[ \sum_{e\in E(\Gamma)} \oplus \begin{bmatrix} 0& 1 \\ 1 & 0 \end{bmatrix}=R_1. \]
\begin{figure}[htbp]
\begin{center}
	\includegraphics[width=100mm]{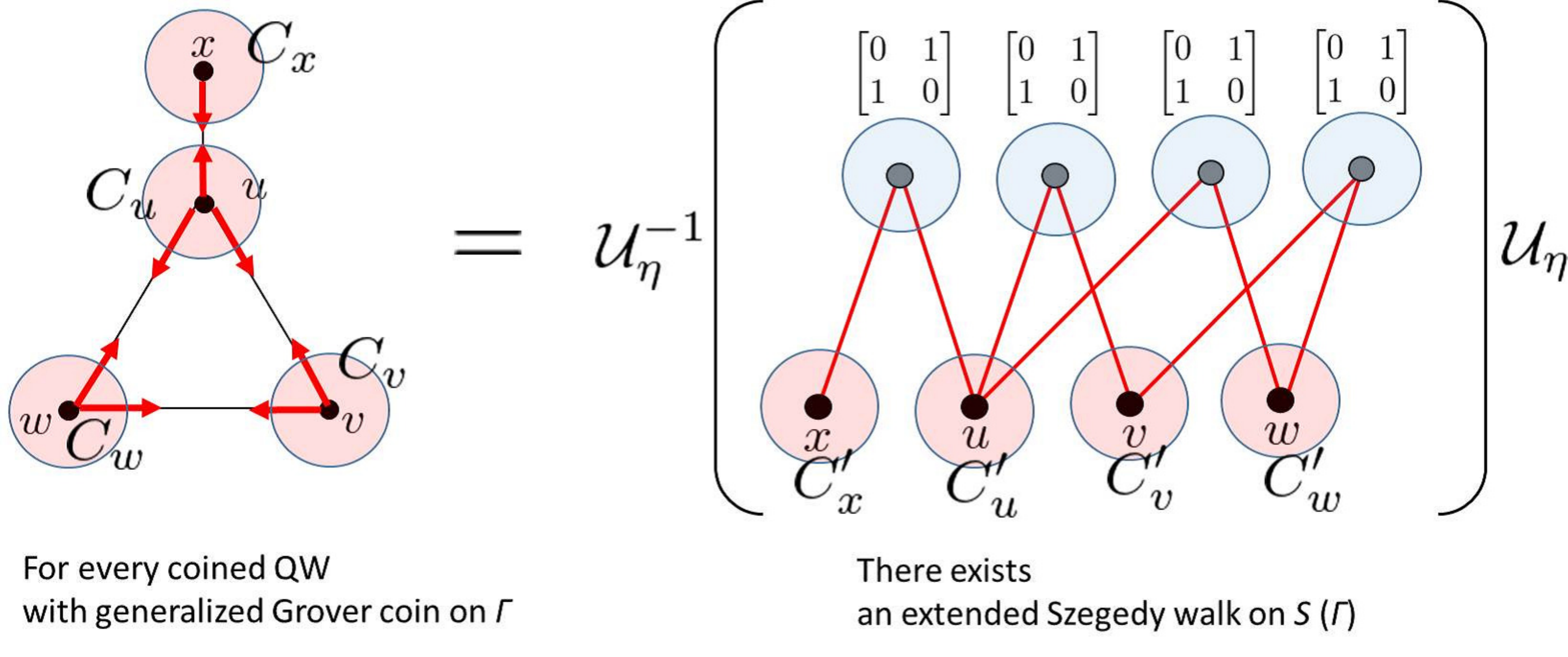}
\end{center}
\caption{An example showing an application of Theorem~\ref{theo1}. }
\label{fig:example}
\end{figure}

\section{Conclusions}\label{sec_conlusions}

We have described a method to convert discrete-time coined QWs that employ the generalized Grover coin on a graph $\Gamma$ into Szegedy's QWs on the subdivision graph of $\Gamma$. This method shows that the internal space of the coined model can be eliminated by including extra vertices into the graph. If the graph on which the coined model is defined is not the complete graph, then the dimension of the Hilbert space of the equivalent Szegedy's model is larger than the one employed by the coined model. This is consequence of the following fact: If Szegedy's model is defined on a bipartite graph with $m$ vertices in the first set and $n$ vertices in the second set, the Hilbert space of Szegedy's model is spanned by $mn$ vectors, but the dynamics takes place in the subspace spanned by the edges of the graph. The number of edges is smaller than $mn$ if the bipartite graph is not complete. Szegedy's model has an idle subspace when defined on non-complete bipartite graphs.

\section*{Acknowledgments}

RP acknowledges financial support from Faperj (grant n.~E-26/102.350/2013) and CNPq (grants n.~303406/2015-1, 4741\-43/2013-9).
ES was partially supported by JSPS Grant-in-Aid for Young Scientists (B) (No. 16K17637) and Japan-Korea Basic Scientific Cooperation Program ``Non-commutative Stochastic Analysis; New Aspects of Quantum White Noise and Quantum Walk'' (2015-2016). The authors thank Iwao Sato, Kaname Matsue, and Raqueline Santos for insightful discussions during the Workshop of Quantum Simulation and Quantum Walks at Yokohama.

\end{document}